\documentclass[aps,pra,twocolumn,amsmath,amssymb,floatfix]{revtex4-2}
\usepackage{graphicx}
\usepackage{physics}
\usepackage{xcolor}
\usepackage{amsthm}
\usepackage{algpseudocode}
\usepackage{pgfplots}
\pgfplotsset{compat=1.18}

\newtheorem{theorem}{Theorem}
\newtheorem{lemma}[theorem]{Lemma}
\newtheorem{proposition}[theorem]{Proposition}
\newtheorem{corollary}[theorem]{Corollary}
\theoremstyle{definition}
\newtheorem{protocol}{Protocol}

\newcounter{algorithm}
\renewcommand{\thealgorithm}{\arabic{algorithm}}

\begin{document}

\title{Assisted quantum teleportation}

\author{Mithilesh Kumar}
\email{mithilesh.kumar@krea.edu.in}
\author{Kaavya Iyer}
\email{kaavya_iyer.sias23@krea.ac.in}
\affiliation{Krea University, Sri City, 517646, India}

\begin{abstract}
Teleportation through a non-maximally entangled pair, e.g., $\ket{\psi(\theta)}_{AB}=\cos\theta\ket{00}+\sin\theta\ket{11}$, induces a noisy channel and cannot achieve deterministic unit-fidelity transmission unless $\theta=\pi/4$. We introduce a framework of \emph{assisted quantum teleportation} in which a third party (the ``Bank'') supplies auxiliary multipartite entanglement to restore a perfect Bell pair on the original $AB$ registers. We analyze two operational roles for the Bank: a Bank-measures model (measurement and broadcast) and a transfer model (the Bank transfers its subsystem and then leaves). For GHZ-class and W-class assistance we derive explicit feasibility regions for deterministic restoration and show an operational inequivalence for W resources. We further characterize finite-shot optimal success probabilities for probabilistic restoration and formulate Bank-measures feasibility for general pure Bank resources as a minimax optimization.
\end{abstract}

\maketitle

\section{Introduction}

Quantum teleportation is a cornerstone of quantum information theory, enabling the transfer of an unknown quantum state between distant parties using a shared entangled resource and classical communication \cite{bennett1993teleporting}. In the standard protocol, Alice and Bob share a maximally entangled two-qubit state, such as the Bell state $\ket{\Phi^+}=\tfrac{1}{\sqrt{2}}(\ket{00}+\ket{11})$. A third-party controller can be granted authority to enable or block teleportation by holding a share of a tripartite resource (controlled teleportation) \cite{karlsson1998}. In that setting the underlying Alice--Bob link is typically assumed ideal, and the controller functions as a switch.

As the field advances toward the quantum internet \cite{wehner2018quantum}, realistic distributed entanglement will be imperfect and often generated via intermediate network elements (e.g., repeaters and swapping-based architectures) \cite{briegel1998quantum,zukowski1993event}. This motivates the complementary setting we study: the primary Alice--Bob channel is non-ideal, and a centralized node (the ``Bank'') provides additional resources to ``repair'' the link.

A basic resource-theoretic question is whether perfect teleportation can be achieved starting from a sub-optimal shared pair, specifically a non-maximally entangled two-qubit state
\begin{equation}\label{eq:betatheta}
    \ket{\psi(\theta)}=\cos\theta\ket{00}+\sin\theta\ket{11}
\end{equation}
with $\theta\in(0,\pi/2)$. Unless $\theta=\pi/4$, deterministic unit-fidelity teleportation is impossible \cite{bennett1996concentrating,horodecki1999general}: for a general two-qubit resource $\rho_{AB}$ the optimal fidelity is bounded by its maximal singlet fraction $F$ as $\mathcal{F}_{\max}=(2F+1)/3$ \cite{horodecki1999general}, and $F=1$ occurs only for states locally equivalent to a Bell pair \cite{verstraete2003optimal}.

We revisit this limitation from a geometric viewpoint. Single-qubit channels act affinely on the Bloch ball, mapping the unit sphere to an ellipsoid \cite{NielsenChuang2000,ruskai2002analysis}. In particular, the standard Bell-measurement teleportation protocol with a non-maximally entangled pure resource deforms the set of teleportable states from a sphere to a prolate spheroid \cite{bowen2001teleportation}, providing a transparent obstruction to deterministic restoration of spherical symmetry.

To overcome this limitation, we establish a framework of \emph{assisted quantum teleportation} in which a Bank provides auxiliary multipartite entanglement to enable Bell-pair restoration on the original registers. We show that the answer can depend on the Bank's operational role and complement deterministic feasibility with finite-shot optimal success probabilities when restoration must be probabilistic.

\paragraph{Relation to assisted entanglement.}
Our setting is conceptually close to ``assistance'' notions in multipartite entanglement theory, where a helper performs local measurements to maximize the entanglement obtainable between two designated parties (entanglement of assistance) \cite{divincenzo1998assistance}, and to network notions such as localizable entanglement \cite{popp2005localizable}. We differ in two operational respects that are central to the present results: (i) the target is a \emph{specific} repaired edge $AB$ (a perfect Bell pair on the original registers), and (ii) we compare two natural Bank roles---the Bank-measures model versus the transfer model---and show they can be inequivalent for W-class resources. We focus on finite-shot (single-copy) feasibility and success probabilities, complementing asymptotic assisted-distillation settings \cite{yard2011assisted}.

\subsection*{Operational models}
We study a fixed imperfect Alice--Bob link $AB$ together with an auxiliary tripartite resource on registers $A'B'K$ initially shared between Alice, Bob, and the Bank. Unless stated otherwise, Alice and Bob may perform arbitrary LOCC on their local holdings, and the Bank participates only through one of the following operational roles:
(i) \emph{Bank-measures model:} the Bank measures $K$ and broadcasts a classical outcome; or
(ii) \emph{transfer model:} the Bank transfers $K$ to one end node (Alice), after which the Bank is absent and Alice and Bob act by LOCC.
Our success criterion is deterministic production of a perfect Bell pair on the original $AB$ registers.

\subsection*{Contributions}
\begin{enumerate}
\item We give a geometric formulation of the impossibility of deterministic unit-fidelity teleportation through $\ket{\psi(\theta)}$ by identifying the induced Bloch-ball contraction to a prolate spheroid.
\item For GHZ-class and W-class Bank resources we derive explicit feasibility conditions for deterministic Bell-pair restoration, and we show an operational inequivalence for W-class resources between the Bank-measures and transfer models.
\item We characterize finite-shot optimal success probabilities for probabilistic Bell-pair restoration and illustrate the resulting performance curves.
\item For general pure Bank resources, we formulate Bank-measures feasibility as a minimax optimization over measurements on $K$, enabling quantitative resource-tradeoff questions.
\end{enumerate}

\section{Entanglement must be consumed for teleportation}
It is well understood that the standard teleportation protocol consumes entanglement in the shared Bell pair \cite{bennett1993teleporting,NielsenChuang2000}. The question we address is whether every teleportation protocol must consume entanglement in the shared resource. The necessity of entanglement consumption in any teleportation protocol can be established by considering the monotonicity of entanglement under local operations and classical communication (LOCC) \cite{vidal1999entanglement,plenio2007introduction,horodecki2009entanglement}. To demonstrate this without assuming a specific protocol or maximal entanglement, we define a setup involving three parties: a reference system $R$, a source qubit $Q$, and a shared resource pair $AB$ distributed between Alice and Bob. Initially, Alice holds $R, Q$, and $A$, while Bob holds $B$ and a target ancilla $T$. Let $E(\rho_{\mathcal{A}|\mathcal{B}})$ denote a valid entanglement measure across the Alice-Bob partition. In the initial state, the total entanglement across the partition is determined strictly by the pre-shared resource:
\begin{equation}E_{\text{initial}}(\mathcal{A}|\mathcal{B}) = E(\rho_{AB})\end{equation}
\begin{figure*}
    \includegraphics[width=0.7\textwidth]{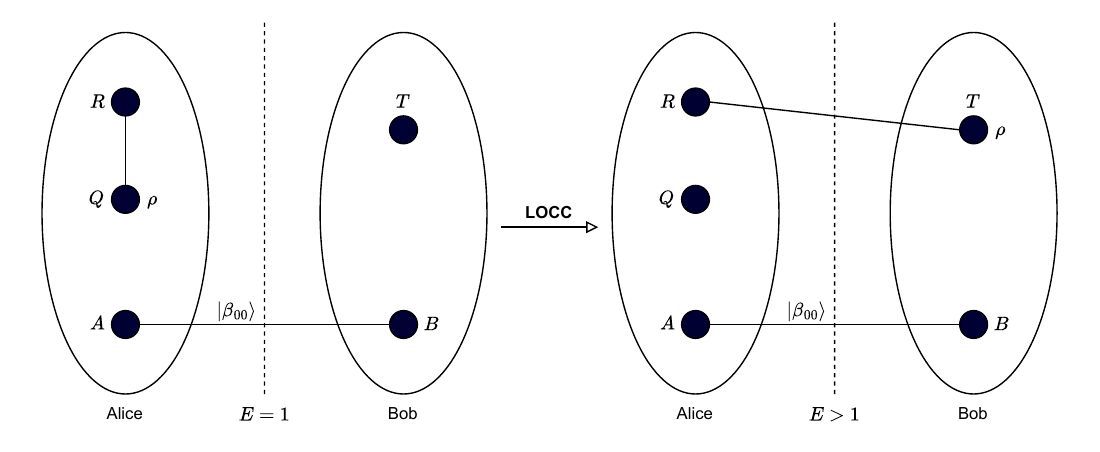}
    \caption{Schematic of the teleportation protocol with a reference system $R$, a source qubit $Q$, and a shared resource pair $AB$ distributed between Alice and Bob. Alice holds $R, Q$, and $A$, while Bob holds $B$ and a target ancilla $T$.}\label{fig:locc}   
\end{figure*}
Note that any entanglement between $R$ and $Q$ ($E_{RQ}$) is local to Alice and thus contributes zero to the cross-partition entanglement. The objective of any general teleportation protocol is to map the state of $Q$ onto $T$. If $Q$ is initially correlated with $R$, a successful unit-fidelity transfer implies that the final state must satisfy $E(\rho_{RT}) = E(\rho_{RQ})$. Suppose there exists a hypothetical protocol that preserves or restores the resource entanglement $E(\rho_{AB})$ while successfully completing the state transfer. In this scenario, the final entanglement across the partition would be:
\begin{equation}E_{\text{final}}(\mathcal{A}|\mathcal{B}) = E(\rho_{AB}) + E(\rho_{RT}) = E(\rho_{AB}) + E(\rho_{RQ})\end{equation}
According to the principle of entanglement monotonicity, entanglement cannot increase under LOCC, requiring $E_{\text{final}} \le E_{\text{initial}}$ \cite{vidal1999entanglement,horodecki2009entanglement}. Applying this to our persistent-resource assumption yields:
\begin{equation}E(\rho_{AB}) + E(\rho_{RQ}) \le E(\rho_{AB})\end{equation}
For any $E(\rho_{RQ}) > 0$, this leads to a direct contradiction. Consequently, the shared entanglement $E(\rho_{AB})$ must be reduced by at least the amount of entanglement transferred to Bob \cite{horodecki2009entanglement,plenio2007introduction}. 

\section{Background: Impossibility of Perfect Teleportation via non-ideal resource}
Standard protocols assume access to an ideal Bell pair. A natural question arises: is it possible to achieve perfect teleportation using a non-ideal resource like $\ket{\psi(\theta)}$? It is well established that deterministic unit-fidelity teleportation is impossible without a maximally entangled resource \cite{horodecki1999general,verstraete2003optimal} (though probabilistic schemes exist \cite{agrawal2002probabilistic}). As this work's primary concern, we merely present a geometric formulation.


The state of a single qubit can be represented by a Bloch vector $\vec v\in\mathbb{R}^3$ via
\begin{equation}
\rho = \frac{1}{2}\bigl(\mathbb{I}+\vec v\cdot\vec\sigma\bigr)
\end{equation}
where $\vec\sigma=(\sigma_x,\sigma_y,\sigma_z)$. Completely positive trace-preserving (CPTP) maps on a single qubit act affinely on the Bloch ball, sending the unit sphere to an ellipsoid (possibly shifted) inside the ball \cite{NielsenChuang2000,ruskai2002analysis}.

In the standard teleportation protocol, using a shared resource state $\rho_{AB}$ and the usual Bell measurement / Pauli correction, Bob receives an effective single-qubit channel $\mathcal{E}_{\rho_{AB}}$ acting on the input state \cite{horodecki1999general,bowen2001teleportation}. In particular, for Bell-diagonal correlations the action on Bloch vectors is linear, $\vec v_{\rm out}=T\,\vec v_{\rm in}$, where $T_{ij}=\mathrm{Tr}(\rho_{AB}\,\sigma_i\otimes\sigma_j)$.

For the non-maximally entangled pure resource $\ket{\psi(\theta)}=\cos\theta\ket{00}+\sin\theta\ket{11}$ one finds
\begin{equation}
T(\theta)=\mathrm{diag}\bigl(\sin(2\theta),\,-\sin(2\theta),\,1\bigr)
\end{equation}
and Bob's Pauli corrections flip the sign of the $y$-axis, yielding the effective contraction
\begin{equation}
\tilde T(\theta)=\mathrm{diag}\bigl(\sin(2\theta),\,\sin(2\theta),\,1\bigr)
\end{equation}
Hence the image of the Bloch ball is the prolate spheroid
\begin{equation}
\frac{(v'_x)^2}{\sin^2(2\theta)}+\frac{(v'_y)^2}{\sin^2(2\theta)}+(v'_z)^2\le 1
\end{equation}
with equatorial semi-axes $a=b=\sin(2\theta)$ and polar semi-axis $c=1$.

\begin{theorem}\label{thm:geometric_impossibility}
For $\theta\neq \pi/4$, perfect deterministic teleportation of an arbitrary qubit state is impossible using $\ket{\psi(\theta)}_{AB}$ as the only shared bipartite resource.
\end{theorem}
\begin{proof}
Perfect deterministic teleportation of an arbitrary input qubit means that, for every input density operator $\rho$, Bob's output equals $\rho$ up to a fixed unitary correction. Equivalently, the induced single-qubit channel must be unitary, and therefore it must map the Bloch ball bijectively onto itself \cite{NielsenChuang2000,ruskai2002analysis}.

For the standard Bell-measurement protocol with resource $\ket{\psi(\theta)}$, the induced channel is unital and acts linearly on Bloch vectors as $\vec v_{\rm out}=\tilde T(\theta)\vec v_{\rm in}$ with $\tilde T(\theta)=\mathrm{diag}(\sin(2\theta),\sin(2\theta),1)$. Thus the image of the Bloch ball is an ellipsoid whose volume is scaled by $\det(\tilde T(\theta))$ relative to the Bloch-ball volume $4\pi/3$, namely
\begin{equation}
V(\theta)=\det\!\bigl(\tilde T(\theta)\bigr)\,\frac{4\pi}{3}=\frac{4\pi}{3}\sin^2(2\theta)
\end{equation}
For $\theta\neq\pi/4$ one has $\sin(2\theta)<1$, so $V(\theta)<4\pi/3$ and the channel is not bijective on the Bloch ball. Hence it cannot be unitary, and therefore cannot implement perfect deterministic teleportation for all input states.
\end{proof}

\paragraph{Remark (singlet-fraction bound).}
Independently of the geometric picture, the optimal deterministic teleportation fidelity achievable from a two-qubit resource $\rho_{AB}$ is bounded by its maximal singlet fraction $F$ as $\mathcal{F}_{\max}=(2F+1)/3$ \cite{horodecki1999general,verstraete2003optimal}. For $\ket{\psi(\theta)}$ with $\theta\neq\pi/4$ one has $F<1$, which also rules out unit-fidelity deterministic teleportation.

\section{Bank-Assisted Multi-Party Teleportation}

We now analyze how a Bank can \emph{restore a Bell pair on $AB$} by distributing auxiliary tripartite entanglement, in the Bank-measures and transfer models defined in the Introduction \cite{wehner2018quantum,briegel1998quantum,zukowski1993event}.

\begin{lemma}[General majorization obstruction]\label{lem:largest_schmidt_obstruction}
Consider any pure global state shared across the bipartition $\mathcal{A}\mid\mathcal{B}$, where $\mathcal{A}$ contains Alice's local registers and $\mathcal{B}$ contains Bob's. If Alice and Bob can deterministically produce a Bell pair on registers $AB$ by LOCC (possibly discarding ancillary systems), then the largest squared Schmidt coefficient of the initial state across $\mathcal{A}\mid\mathcal{B}$ must be at most $1/2$ \cite{nielsen1999conditions,marshall2011inequalities}.
\end{lemma}

\subsection{Generalized GHZ-class resource: deterministic compensation in the Bank-measures model}
\label{sec:ghz_assisted}

The Bank distributes a generalized GHZ-class state \cite{dur2000three}
\begin{equation}\label{eq:ghz-gen}
    \ket{\Phi}_{A'B'K}=\alpha\ket{000}+\beta\ket{111}
\end{equation}
with $\alpha\ge\beta>0$ and $\alpha^2+\beta^2=1$, sending $A'$ to Alice, $B'$ to Bob, and retaining $K$. In the Bank-measures model, the Bank measures $K$ and broadcasts one bit, selecting a known conditional state for Alice and Bob (cf. controlled teleportation) \cite{karlsson1998}.

\begin{protocol}[GHZ-assisted Bell-pair restoration in the Bank-measures model]\label{prot:ghz}
Define the global initial state
\begin{equation}
\ket{\Psi_0}=(\cos\theta\ket{00}+\sin\theta\ket{11})_{AB}\otimes(\alpha\ket{000}+\beta\ket{111})_{A'B'K}
\end{equation}
Expanding in the Hadamard basis of $K$ gives
\begin{align}
\ket{\Psi_0}=\frac{1}{\sqrt{2}}\bigl(\ket{\hat\Psi_+}\ket{+}_K+\ket{\hat\Psi_-}\ket{-}_K\bigr)
\end{align}
where
\begin{align}
\ket{\hat\Psi_{\pm}}=(\cos\theta\ket{00}+\sin\theta\ket{11})_{AB}\otimes(\alpha\ket{00}\pm\beta\ket{11})_{A'B'}
\end{align}
and $\ket{\hat\Psi_-}=(\mathbb{I}_{AA'B}\otimes\sigma_z^{B'})\ket{\hat\Psi_+}$.

By Nielsen's theorem, deterministic restoration of a Bell pair on $AB$ by LOCC across $AA'\mid BB'$ is possible iff the Schmidt vector of $\ket{\hat\Psi_+}$ is majorized by $(1/2,1/2,0,0)^T$ \cite{nielsen1999conditions,marshall2011inequalities}. The squared Schmidt coefficients are
\begin{equation}
\vec\lambda=\bigl(\alpha^2\cos^2\!\theta,\,\alpha^2\sin^2\!\theta,\,\beta^2\cos^2\!\theta,\,\beta^2\sin^2\!\theta\bigr)^T
\end{equation}
which yields the following feasibility condition.

\begin{theorem}[GHZ-assisted deterministic restoration]\label{thm:ghz_feasible}
Fix $\theta\in(0,\pi/4]$ and let the Bank distribute $\ket{\Phi}_{A'B'K}=\alpha\ket{000}+\beta\ket{111}$ with $\alpha\ge\beta>0$. In the Bank-measures model (Bank measures $K$ and broadcasts one bit), Alice and Bob can deterministically restore a perfect Bell pair on $AB$ by LOCC if and only if
\begin{equation}
\alpha^2\le \frac{1}{2\cos^2\!\theta}
\end{equation}
The same feasibility condition holds in the transfer model when $K$ is moved to Alice.
\end{theorem}

\paragraph{Minimal Bank entanglement within the GHZ family.}
Across the bipartition $K\mid A'B'$ the GHZ resource has entropy of entanglement
\begin{equation}
E_{K: A'B'}(\Phi)=h_2(\alpha^2)
\end{equation}
where $h_2(p)=-p\log_2 p-(1-p)\log_2(1-p)$.
Within the restricted GHZ family and for fixed $\theta$, Theorem~\ref{thm:ghz_feasible} implies that the least-entangled feasible GHZ resource is obtained by saturating the bound, $\alpha^2=1/(2\cos^2\!\theta)$.

\subsection{Finite-shot optimal success probability}
When the deterministic condition in Theorem~\ref{thm:ghz_feasible} fails, Alice and Bob may still probabilistically restore a Bell pair. For pure-state transformations under LOCC, Vidal's theorem gives the optimal success probability \cite{vidal1999entanglement}. Specializing Vidal's expression to Bell-pair restoration yields the closed form
\begin{equation}
P_{\max}=\min\bigl(1,2(1-\lambda_{\max})\bigr)
\end{equation}
(see Appendix~\ref{app:vidal}), where $\lambda_{\max}$ is the largest squared Schmidt coefficient of the initial pure state across the relevant Alice--Bob bipartition.

For GHZ assistance, across $AA'\mid BB'$ (Bank-measures model) one has $\lambda_{\max}=\alpha^2\cos^2\!\theta$, hence
\begin{equation}
P_{\max}^{\rm GHZ}(\theta,\alpha)=\min\bigl(1,2(1-\alpha^2\cos^2\!\theta)\bigr)
\end{equation}
The same expression holds in the transfer model (across $AA'K\mid BB'$) since the Schmidt spectrum is unchanged by relocating $K$.

\begin{figure}[t]
\centering
\begin{tikzpicture}
\begin{axis}[
  width=\columnwidth,
  height=0.7\columnwidth,
  xmin=0.5, xmax=1,
  ymin=0, ymax=1.05,
  xlabel={$\alpha^2$},
  ylabel={$P_{\max}^{\rm GHZ}$},
  legend style={at={(0.02,0.02)},anchor=south west,font=\small},
  domain=0.5:1,
  samples=200,
]
\addplot[blue,thick] {min(1,2*(1-x*(cos(deg(pi/8)))^2))};
\addlegendentry{$\theta=\pi/8$}
\addplot[red,dashed,thick] {min(1,2*(1-x*(cos(deg(pi/6)))^2))};
\addlegendentry{$\theta=\pi/6$}
\addplot[black,dotted,thick] {min(1,2*(1-x*(cos(deg(pi/5)))^2))};
\addlegendentry{$\theta=\pi/5$}
\end{axis}
\end{tikzpicture}
\caption{Finite-shot optimal success probability for GHZ-assisted Bell-pair restoration as a function of $\alpha^2$ for representative link angles $\theta$. The curves satisfy $P_{\max}^{\rm GHZ}(\theta,\alpha)=\min\bigl(1,2(1-\alpha^2\cos^2\!\theta)\bigr)$.}
\label{fig:pmax_ghz}
\end{figure}
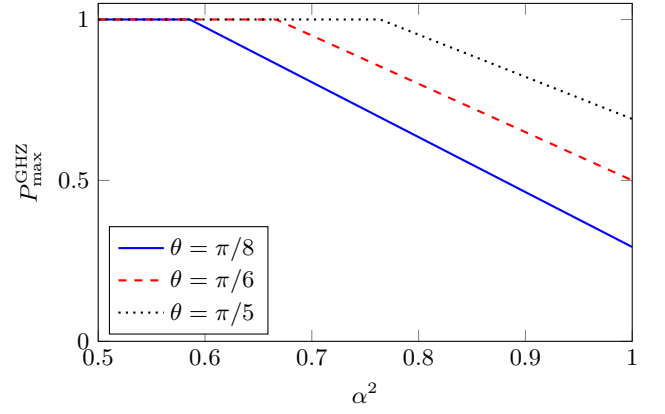

When the condition holds, an explicit deterministic LOCC map can be implemented via a finite-outcome POVM whose construction follows from a Birkhoff--von Neumann decomposition of the relevant doubly-stochastic map \cite{birkhoff1946tres} (we give the algorithm in Appendix~\ref{app:birkhoff}). Concretely, Alice measures $AA'$ with operators
\begin{equation}
M_k=\sqrt{p_k}\,\Lambda_{\text{target}}^{1/2}P_k^T\Lambda_{\text{source}}^{-1/2}
\end{equation}
communicates $k$ to Bob, and Bob applies the corresponding permutation unitary $U_k=P_k^T$ on $BB'$ to obtain $\ket{\Phi^+}_{AB}$.

If Alice and Bob apply this LOCC map before the Bank measures $K$, the post-processing produces a coherent superposition correlated with $K$,
\begin{equation}
\ket{\Psi_{\mathrm{final}}}=\frac{1}{\sqrt{2}}\bigl(\ket{\Phi^+}_{AB}\ket{+}_K+\tilde Z_k\ket{\Phi^+}_{AB}\ket{-}_K\bigr)
\end{equation}
where $\tilde Z_k=U_k(\mathbb{I}_B\otimes\sigma_z^{B'})U_k^{\dagger}$.
Tracing out $K$ yields a dephased mixture on $AB$, so the Bank must eventually measure $K$ in the $\{\ket{\pm}\}$ basis and send one classical bit to complete the correction (Bank-measures model).
\end{protocol}

\subsection{Asymmetric distribution: Bank transfers $K$ to Alice}

If the Bank transfers qubit $K$ to Alice, the active third party disappears, but the LOCC feasibility condition is unchanged (Schmidt spectra are invariant under local relocation). Across the cut $AA'K\mid BB'$ the state can be written
\begin{equation}
\begin{aligned}
\ket{\Psi_0} &= \alpha\cos\theta\ket{000}_{AA'K}\ket{00}_{BB'}\\ &+\beta\cos\theta\ket{011}_{AA'K}\ket{01}_{BB'}\\
&\quad+\alpha\sin\theta\ket{100}_{AA'K}\ket{10}_{BB'}\\ &+\beta\sin\theta\ket{111}_{AA'K}\ket{11}_{BB'}
\end{aligned}
\end{equation}
with squared Schmidt coefficients $(\alpha^2\cos^2\!\theta,\beta^2\cos^2\!\theta,\alpha^2\sin^2\!\theta,\beta^2\sin^2\!\theta)$.

Operationally, Alice can now (i) measure $K$ in the Hadamard basis herself to learn the $\pm$ branch and then run the same majorization POVM on $AA'$; or (ii) absorb the $K$-measurement into a single joint POVM on $AA'K$ by taking tensor products of the $M_k$ with $\ket{\pm}\!\bra{\pm}_K$. In either case, Alice sends Bob the required classical indices and Bob applies the corresponding unitary corrections, recovering $\ket{\Phi^+}_{AB}$ deterministically without any Bank-side measurement.

\section{W-Class States}

A second natural choice for a Bank-supplied tripartite resource is a W-class state \cite{dur2000three}
\begin{equation}
\ket{W}_{A'B'K}=\alpha\ket{001}+\beta\ket{010}+\gamma\ket{100}
\end{equation}
with $\alpha,\beta,\gamma>0$ and $\alpha^2+\beta^2+\gamma^2=1$. If the Bank measures $K$ in the Hadamard basis and broadcasts the one-bit outcome, Alice and Bob share the (known) conditional state
\begin{equation}
\ket{w_{\pm}}_{A'B'}=\beta\ket{01}+\gamma\ket{10}\pm\alpha\ket{00}
\end{equation}
so that the post-measurement branch on $AA'BB'$ is
\begin{equation}\label{eq:W-branch}
\ket{\Psi_{\pm}}_{AA'BB'}=(\cos\theta\ket{00}+\sin\theta\ket{11})_{AB}\otimes\ket{w_{\pm}}_{A'B'}
\end{equation}

Deterministic recovery of a Bell pair on $AB$ reduces to an LOCC conversion across the cut $AA'\mid BB'$. By Nielsen's theorem (majorization) \cite{nielsen1999conditions,marshall2011inequalities}, this is possible iff the Schmidt vector of $\ket{\Psi_{\pm}}$ is majorized by $(1/2,1/2,0,0)^T$. Diagonalizing the reduced state on $BB'$ gives the squared Schmidt coefficients
\begin{align}
\lambda_{1,2} &= \cos^2\!\theta\left(\frac{1\pm\sqrt{1-4\beta^2\gamma^2}}{2}\right)\\
\lambda_{3,4} &= \sin^2\!\theta\left(\frac{1\pm\sqrt{1-4\beta^2\gamma^2}}{2}\right)
\end{align}
so the binding constraint is $\lambda_1\le 1/2$, i.e.
\begin{equation}\label{eq:W-majorization}
\cos^2\!\theta\bigl(1+\sqrt{1-4\beta^2\gamma^2}\bigr)\le 1
\end{equation}
Rearranging yields the convenient feasibility condition
\begin{equation}\label{eq:W-constraint}
4\beta^2\gamma^2\ge 1-\tan^4\!\theta
\end{equation}

\begin{theorem}[W-assisted deterministic restoration in the Bank-measures model]\label{thm:w_bankmeasures_feasible}
Fix $\theta\in(0,\pi/4]$ and let the Bank distribute $\ket{W}_{A'B'K}=\alpha\ket{001}+\beta\ket{010}+\gamma\ket{100}$ with $\alpha,\beta,\gamma>0$. In the Bank-measures model, Alice and Bob can deterministically restore a Bell pair on $AB$ by LOCC if and only if Eq.~\eqref{eq:W-constraint} holds.
\end{theorem}

\subsection{Finite-shot optimal success probability}
When Eq.~\eqref{eq:W-constraint} fails, deterministic restoration is impossible in the Bank-measures model, but probabilistic restoration may still succeed. Since each post-measurement branch is a pure state on $AA'\mid BB'$, Vidal's theorem applies directly \cite{vidal1999entanglement}, with the Bell-pair specialization summarized in Appendix~\ref{app:vidal}. Using the largest squared Schmidt coefficient $\lambda_1$ from Eq.~\eqref{eq:W-majorization}, the optimal probability to restore a Bell pair in a branch is
\begin{equation}
P_{\max}^{\rm W,meas}(\theta,\beta,\gamma)=\min\bigl(1,2(1-\lambda_1)\bigr)
\end{equation}
Equivalently, writing $\Delta=\sqrt{1-4\beta^2\gamma^2}$ so that $\lambda_1=\tfrac{\cos^2\!\theta}{2}(1+\Delta)$, one obtains
\begin{equation}
P_{\max}^{\rm W,meas}(\theta,\beta,\gamma)=\min\Bigl(1,2-\cos^2\!\theta(1+\Delta)\Bigr)
\end{equation}
In particular, $P_{\max}^{\rm W,meas}=1$ exactly when Eq.~\eqref{eq:W-constraint} holds.

In the transfer model, the pure state across $AA'K\mid BB'$ has largest squared Schmidt coefficient
\begin{equation}
\lambda_{\max}=\cos^2\!\theta\max\bigl(\beta^2,1-\beta^2\bigr)
\end{equation}
and therefore
\begin{equation}
P_{\max}^{\rm W,trans}(\theta,\beta)=\min\Bigl(1,2\bigl(1-\cos^2\!\theta\max(\beta^2,1-\beta^2)\bigr)\Bigr)
\end{equation}

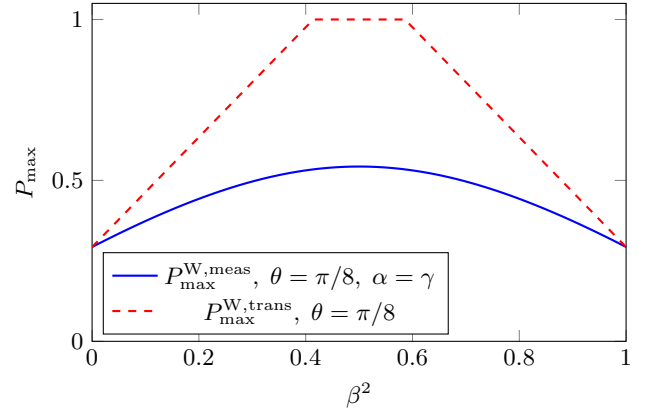
\begin{figure}[t]
\centering
\begin{tikzpicture}
\begin{axis}[
  width=\columnwidth,
  height=0.7\columnwidth,
  xmin=0, xmax=1,
  ymin=0, ymax=1.05,
  xlabel={$\beta^2$},
  ylabel={$P_{\max}$},
  legend style={at={(0.02,0.02)},anchor=south west,font=\small},
  domain=0:1,
  samples=250,
]
\addplot[blue,thick] {min(1, 2 - (cos(deg(pi/8)))^2 * (1 + sqrt(1 - 2*x*(1-x))))};
\addlegendentry{$P_{\max}^{\rm W,meas},\;\theta=\pi/8,\;\alpha=\gamma$}
\addplot[red,dashed,thick] {min(1, 2*(1 - (cos(deg(pi/8)))^2 * max(x,1-x)))};
\addlegendentry{$P_{\max}^{\rm W,trans},\;\theta=\pi/8$}
\end{axis}
\end{tikzpicture}
\caption{Finite-shot optimal success probability for W-assisted Bell-pair restoration as a function of $\beta^2$ at fixed $\theta=\pi/8$. On the symmetric slice $\alpha=\gamma$ (solid), the curve satisfies $P_{\max}^{\rm W,meas}=\min\bigl(1,2-\cos^2\!\theta(1+\sqrt{1-2\beta^2(1-\beta^2)})\bigr)$. In the transfer model (dashed), $P_{\max}^{\rm W,trans}=\min\bigl(1,2(1-\cos^2\!\theta\max(\beta^2,1-\beta^2))\bigr)$.}
\label{fig:pmax_w}
\end{figure}

\subsection{Operational inequivalence: measuring $K$ vs transferring $K$}

Theorem~\ref{thm:w_bankmeasures_feasible} gives the deterministic feasibility condition in the Bank-measures model. We now derive the corresponding condition in the transfer model and compare the two.

\textbf{Transfer model.} After transferring $K$ to Alice, the relevant bipartition is $AA'K\mid BB'$. The squared Schmidt coefficients are
\begin{align}
\lambda_1&=\cos^2\!\theta\,(\alpha^2+\gamma^2),
&
\lambda_2&=\cos^2\!\theta\,\beta^2,\nonumber\\
\lambda_3&=\sin^2\!\theta\,(\alpha^2+\gamma^2),
&
\lambda_4&=\sin^2\!\theta\,\beta^2
\end{align}
so Nielsen majorization reduces to $\max(\lambda_1,\lambda_2)\le 1/2$ \cite{nielsen1999conditions}. Using $\alpha^2+\gamma^2=1-\beta^2$ yields
\begin{equation}\label{eq:w_transfer_bound}
\frac{1-\tan^2\!\theta}{2}\le \beta^2\le \frac{1+\tan^2\!\theta}{2}
\end{equation}

\begin{proposition}[W-assisted deterministic restoration with transfer]\label{prop:w_transfer_feasible}
Fix $\theta\in(0,\pi/4]$ and a W-class resource $\ket{W}_{A'B'K}$. In the transfer model, Alice and Bob can deterministically restore a Bell pair on $AB$ by LOCC if and only if Eq.~\eqref{eq:w_transfer_bound} holds.
\end{proposition}

\paragraph{Separation.}
Comparing Proposition~\ref{prop:w_transfer_feasible} with Theorem~\ref{thm:w_bankmeasures_feasible} shows that the two operational models need not coincide.

\begin{corollary}[Operational separation for W-class assistance]\label{cor:w_separation}
There exist $(\theta,\alpha,\beta,\gamma)$ for which deterministic restoration is possible in the transfer model but impossible in the Bank-measures model.
\end{corollary}

To make the separation explicit on a simple one-parameter slice, assume $\alpha=\gamma$.

\begin{proposition}[Separation region for the symmetric slice $\alpha=\gamma$]\label{prop:w_symmetric_separation}
On the symmetric slice $\alpha=\gamma$ (so $\alpha^2=\gamma^2=(1-\beta^2)/2$), the Bank-measures feasibility condition Eq.~\eqref{eq:W-constraint} is satisfiable for some $\beta^2$ if and only if $\tan^4\!\theta\ge 1/2$. In contrast, transfer feasibility holds for some $\beta^2$ for every $\theta\in(0,\pi/4]$.
\end{proposition}

Figure~\ref{fig:w_state_phase} illustrates this mismatch. Quantitative resource-cost refinements (e.g., minimal $E_{K:A'B'}$ on a slice) are discussed in Section~\ref{sec:extensions}.

\begin{figure}[htbp]
\centering
\includegraphics[width=\columnwidth]{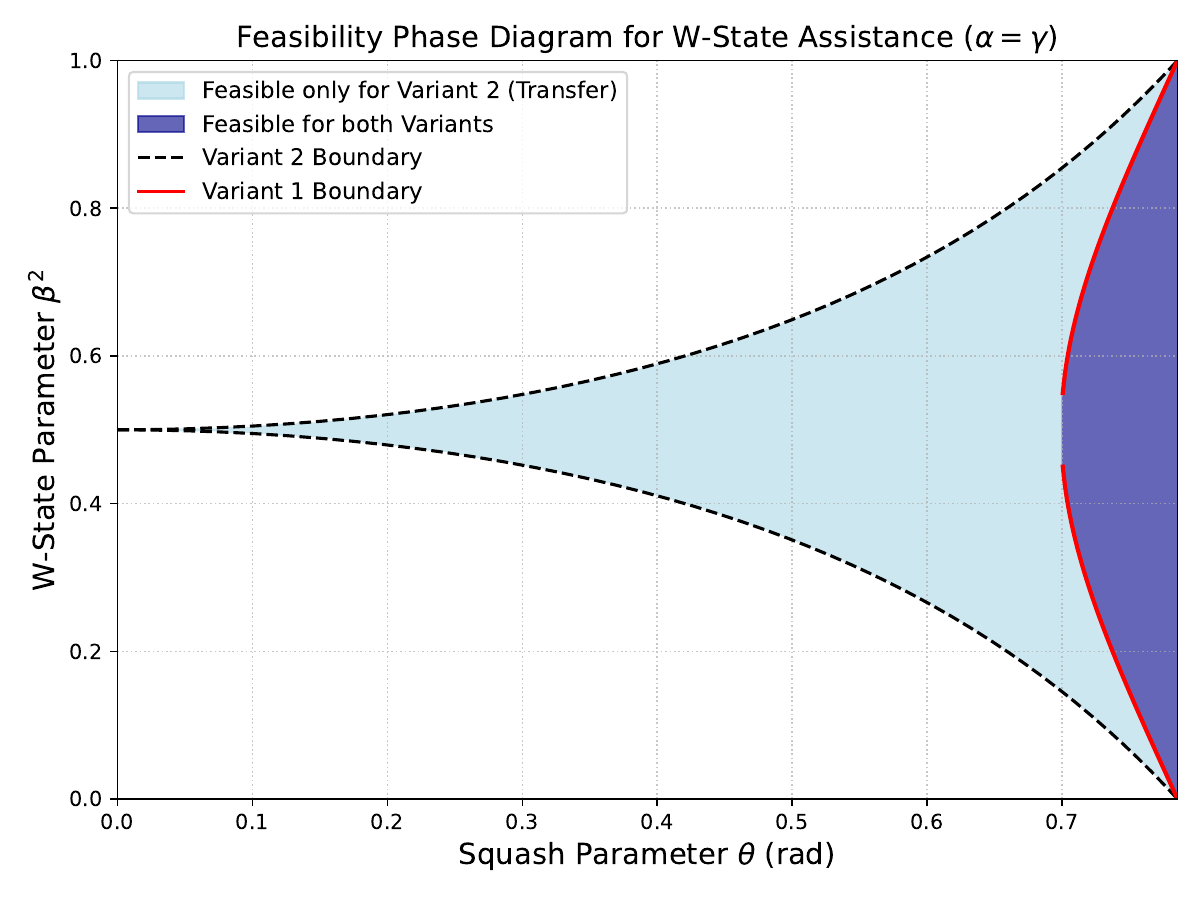}
\caption{Feasibility phase diagram for W-state assistance under the symmetric assumption $\alpha=\gamma$. The light-blue region highlights the operational inequivalence: these parameters succeed only when the Bank transfers $K$ to Alice (not when it measures $K$).}
\label{fig:w_state_phase}
\end{figure}
\section{General Bank Resources Beyond GHZ and W}
\label{sec:general_bank}

The GHZ- and W-class analyses above exploit the fact that the global state is pure across the relevant Alice--Bob bipartition, so deterministic Bell-pair restoration reduces to a majorization test on Schmidt coefficients \cite{nielsen1999conditions,marshall2011inequalities}. The same logic yields a compact characterization for \emph{arbitrary} pure Bank resources on $A'B'K$ in the transfer model, and a natural ``assisted'' optimization formulation in the Bank-measures model.

\subsection{Transfer model: a universal majorization criterion}
Let the imperfect link be $\ket{\psi(\theta)}_{AB}$ and let the Bank distribute an arbitrary pure state $\ket{\Omega}_{A'B'K}$. After transferring $K$ to Alice, the joint state is pure across the bipartition $\mathcal{A}\mid\mathcal{B}$ with $\mathcal{A}=AA'K$ and $\mathcal{B}=BB'$. Alice and Bob seek deterministic Bell-pair restoration on $AB$ by LOCC, allowing arbitrary leftover ``junk'' registers.

\begin{theorem}[General transfer-model feasibility]
\label{thm:general_transfer}
In the transfer model, deterministic Bell-pair restoration on $AB$ is possible by LOCC if and only if the largest squared Schmidt coefficient of the initial pure state across $AA'K\mid BB'$ is at most $1/2$.
\end{theorem}

\begin{proof}
The necessity of $\lambda_{\max}\le 1/2$ follows from Lemma~\ref{lem:largest_schmidt_obstruction}. For sufficiency, let $\vec\lambda$ be the squared Schmidt vector of the initial pure state across $AA'K\mid BB'$ and note that the target state $\ket{\Phi^+}_{AB}\otimes\ket{\mathrm{junk}}$ has squared Schmidt vector $(1/2,1/2,0,\dots)$ across the same cut. The condition $\lambda_{\max}\le 1/2$ is equivalent to the majorization relation $\vec\lambda\prec(1/2,1/2,0,\dots)$, and therefore the conversion is possible by Nielsen's theorem \cite{nielsen1999conditions,marshall2011inequalities}.
\end{proof}

Theorem~\ref{thm:general_transfer} recovers the earlier GHZ transfer equivalence (Theorem~\ref{thm:ghz_feasible}) and the W transfer condition (Proposition~\ref{prop:w_transfer_feasible}) as special cases.

\paragraph{Constructive protocol (transfer model).}
When $\lambda_{\max}\le 1/2$, an explicit finite-outcome LOCC protocol can be obtained from the standard constructive proof of Nielsen's theorem via a Birkhoff--von Neumann decomposition. Since this construction is the same as in the GHZ case, we do not repeat it here; see Appendix~\ref{app:birkhoff} for the decomposition algorithm and the corresponding POVM form.

\paragraph{Finite-shot success probability in the transfer model.}
Vidal's theorem for pure-state LOCC transformations \cite{vidal1999entanglement} implies that the optimal probability of converting an arbitrary pure state with largest squared Schmidt coefficient $\lambda_{\max}$ into a Bell pair on $AB$ (allowing arbitrary leftover systems) is
\begin{equation}
P_{\max}=\min\bigl(1,2(1-\lambda_{\max})\bigr)
\end{equation}
(see Appendix~\ref{app:vidal} for the specialization).

\subsection{Bank-measures model: assisted feasibility as a minimax problem}
In contrast to the transfer model (a single pure-state majorization test), the Bank-measures model introduces a \emph{family} of conditional bipartite states: the Bank chooses a measurement on $K$, but deterministic restoration requires that \emph{every} measurement branch admit Bell-pair restoration after the outcome is broadcast. This naturally leads to a worst-case (minimax) criterion.

Concretely, the Bank performs a measurement on $K$ and broadcasts the classical outcome $m$. For a chosen POVM $\{M_m\}_m$ on $K$ satisfying $\sum_m M_m^{\dagger}M_m=\mathbb{I}_K$, the unnormalized post-measurement pure state on Alice and Bob's registers is
\begin{equation}
\ket{\widetilde\Psi_m}_{AA'BB'}=(\mathbb{I}_{AA'BB'}\otimes M_m)\bigl(\ket{\psi(\theta)}_{AB}\otimes\ket{\Omega}_{A'B'K}\bigr)
\end{equation}
with branch probability $p_m=\braket{\widetilde\Psi_m|\widetilde\Psi_m}$. Conditioning on $m$ gives the normalized state $\ket{\Psi_m}=\ket{\widetilde\Psi_m}/\sqrt{p_m}$. After receiving $m$, Alice and Bob may apply an LOCC map (which may depend on $m$) to restore a Bell pair on $AB$.

For each branch, deterministic Bell-pair restoration is possible iff the largest squared Schmidt coefficient of $\ket{\Psi_m}$ across $AA'\mid BB'$ is at most $1/2$ (the same majorization condition as in Theorem~\ref{thm:general_transfer}). This motivates the assisted minimax quantity
\begin{equation}
\mu_*(\theta,\Omega)=\inf_{\{M_m\}}\sup_m\lambda_{\max}(\Psi_m)
\end{equation}
where $\lambda_{\max}(\Psi_m)$ denotes the largest squared Schmidt coefficient of $\ket{\Psi_m}$ across $AA'\mid BB'$.

\begin{proposition}[Bank-measures feasibility as a minimax condition]
\label{prop:bankmeasures_minimax}
Deterministic Bell-pair restoration in the Bank-measures model is possible if and only if $\mu_*(\theta,\Omega)\le 1/2$
\end{proposition}

This formulation is in the spirit of entanglement-of-assistance and localizable-entanglement optimizations, where a helper chooses a measurement to optimize an entanglement figure of merit \cite{divincenzo1998assistance,popp2005localizable}. In particular, for GHZ-class resources the Hadamard measurement on $K$ achieves the feasibility threshold in Theorem~\ref{thm:ghz_feasible}, whereas for W-class resources the choice of measurement affects whether the bound $1/2$ can be met (Theorem~\ref{thm:w_bankmeasures_feasible}).

\section{Resource Cost Tradeoffs}
\label{sec:extensions}

We briefly discuss quantitative tradeoffs between Bank resources in the Bank-measures and transfer models.

\paragraph{Entanglement cost.}
For a pure Bank state $\ket{\Omega}_{A'B'K}$ a natural ``Bank cost'' is the entanglement across $K\mid A'B'$, namely
\begin{equation}
E_{K:A'B'}(\Omega)=S(\rho_K)
\end{equation}
where $S$ is the von Neumann entropy. Within restricted families (e.g., GHZ) we identified the least-entangled feasible resource by saturating the deterministic feasibility bound. A general optimization suggested by Proposition~\ref{prop:bankmeasures_minimax} is
\begin{equation}
E_*^{\rm meas}(\theta)=\inf\bigl\{E_{K:A'B'}(\Omega):\mu_*(\theta,\Omega)\le 1/2\bigr\}
\end{equation}
(and analogously $E_*^{\rm trans}(\theta)$ with the transfer-model constraint $\lambda_{\max}\le 1/2$). Computing $E_*^{\rm meas}$ and comparing it to $E_*^{\rm trans}$ would quantify the operational inequivalence beyond the special GHZ/W parameterizations.

\paragraph{Classical communication from the Bank.}
In the transfer model the Bank sends no classical message. In the Bank-measures model, if the Bank uses $M$ outcomes then it must broadcast at least $\lceil\log_2 M\rceil$ bits. In particular, if deterministic restoration requires at least two genuinely distinct measurement outcomes (not reducible to a fixed known local unitary correction on Alice and Bob), then at least one classical bit must be broadcast. In our GHZ protocol the Bank-measures message is one bit (the $\pm$ branch), whereas the transfer model can remove this requirement by relocating $K$ to Alice.

\paragraph{Minimal Bank dimension.}
A minimal-dimension question is: what is the smallest $\dim K$ that can enable Bell-pair restoration for a given $\theta$ under each model.
\begin{itemize}
\item \emph{Lower bound.} In the Bank-measures model, $\dim K=1$ implies the Bank has no nontrivial measurement and the problem reduces to ordinary bipartite LOCC on $AB$, which cannot restore a Bell pair unless $\theta=\pi/4$.
\item \emph{Achievability with a qubit Bank register.} For every $\theta\in(0,\pi/4]$, the GHZ construction shows that $\dim K=2$ (a single qubit at the Bank) suffices in the Bank-measures model, and by relocation it also suffices in the transfer model.
\item \emph{Dimension tradeoffs.} One can make the ``minimal Bank dimension'' question precise by fixing local dimension budgets $(d_{A'},d_{B'},d_K)$ and asking whether there exists a pure state $\ket{\Omega}\in\mathbb{C}^{d_{A'}}\otimes\mathbb{C}^{d_{B'}}\otimes\mathbb{C}^{d_K}$ such that deterministic restoration is possible. In the transfer model this is equivalent to the existence of $\ket{\Omega}$ with $\lambda_{\max}(AA'K\mid BB')\le 1/2$ (Theorem~\ref{thm:general_transfer}); in the Bank-measures model it is equivalent to the existence of $\ket{\Omega}$ with $\mu_*(\theta,\Omega)\le 1/2$ (Proposition~\ref{prop:bankmeasures_minimax}).
\end{itemize}

\section{Entanglement Catalysis}
We include entanglement catalysis as a benchmark mechanism for comparison with Bank-supplied tripartite assistance. In entanglement catalysis, an otherwise impossible LOCC conversion $\ket{\psi}\to\ket{\phi}$ becomes possible when Alice and Bob are given an auxiliary entangled state $\ket{C}$ that is returned unchanged at the end of the protocol \cite{jonathan1999entanglement}. In our setting, this corresponds to the Bank distributing $\ket{C}$ and then remaining passive.

We take a two-qubit catalyst of the form
\begin{equation}\label{eq:catalyst}
\ket{C}_{A'B'}=\sqrt{c_1}\ket{00}_{A'B'}+\sqrt{c_2}\ket{11}_{A'B'}
\end{equation}
where $c_1\ge c_2>0,~ c_1+c_2=1$ with $A'$ sent to Alice and $B'$ to Bob.

\subsection{Deterministic catalysis}

For pure-state bipartite conversions, deterministic LOCC feasibility is equivalent to majorization of Schmidt vectors (Nielsen's theorem) \cite{nielsen1999conditions,marshall2011inequalities}. Consider
\[\begin{aligned}
\ket{\Psi_{\mathrm{in}}} &=(\cos\theta\ket{00}+\sin\theta\ket{11})_{AB}\otimes\ket{C}_{A'B'}
\\
\ket{\Psi_{\mathrm{out}}} &=\ket{\Phi^+}_{AB}\otimes\ket{C}_{A'B'}
\end{aligned}\]
Sorting Schmidt coefficients in descending order (assume $c_1\ge c_2$), we have
\begin{align}
\vec\lambda_{\mathrm{in}}&=\bigl(\cos^2\!\theta\,c_1,\,\cos^2\!\theta\,c_2,\,\sin^2\!\theta\,c_1,\,\sin^2\!\theta\,c_2\bigr)^{\downarrow}\\
\vec\lambda_{\mathrm{out}}&=\bigl(c_1/2,\,c_1/2,\,c_2/2,\,c_2/2\bigr)^{\downarrow}
\end{align}
The first (largest-component) majorization inequality forces
\begin{equation}
\cos^2\!\theta\,c_1\le c_1/2\quad\Longleftrightarrow\quad \cos^2\!\theta\le\frac{1}{2}
\end{equation}
which is independent of the catalyst. Hence deterministic catalysis cannot restore a maximally entangled pair when $\theta<\pi/4$.

\subsection{Optimal stochastic conversion (SLOCC) and a matching filter}

When $\theta<\pi/4$, the best one can do is a probabilistic (SLOCC) conversion. Vidal's theorem gives the optimal success probability for pure-state transformations under LOCC \cite{vidal1999entanglement}. For the present pair of states, it yields
\begin{equation}
P_{\max}=2\sin^2\!\theta
\end{equation}

This bound is achievable by a single-sided local filter (the ``Procrustean'' concentration protocol) acting only on Alice's $A$ system, leaving the catalyst untouched \cite{bennett1996concentrating}. One convenient Kraus operator for the success branch is
\begin{equation}
M_{\mathrm{succ}}=\begin{pmatrix}\tan\theta&0\\0&1\end{pmatrix}_{\!A}\otimes\mathbb{I}_{A'}
\end{equation}
with Bob applying the identity. On success,
\begin{equation}
(M_{\mathrm{succ}}\otimes\mathbb{I}_{BB'})\ket{\Psi_{\mathrm{in}}}=\sqrt{2}\sin\theta\,\bigl(\ket{\Phi^+}_{AB}\otimes\ket{C}_{A'B'}\bigr)
\end{equation}
so the success probability is $P=2\sin^2\!\theta$, matching Vidal's upper bound, and the catalyst is returned exactly.

\paragraph{Remark.}
Catalytic convertibility (``trumping'') has complete characterizations in terms of families of inequalities \cite{klimesh2007inequalities,turgut2007catalytic}, but for this particular teleportation-restoration target the simple largest-Schmidt-coefficient obstruction already fixes the deterministic threshold.

\section{Non-Local Operations and Entanglement Routing}
Thus far, we have restricted the Bank to distributing pre-prepared entanglement. As a further comparison point (a strictly stronger capability), we now consider a model in which the Bank can directly implement bipartite unitaries with the end nodes (non-local operations). Operationally, this is a rudimentary model of an ``active'' network element that can coherently interact with, and thereby route entanglement between, distant links \cite{NielsenChuang2000,zukowski1993event}.

\subsection{Single-sided non-local operations}

Assume the Bank can apply an arbitrary joint unitary $U_{KA}$ on its register $K$ and Alice's qubit $A$, but has no non-local access to Bob. Without any pre-shared Bank--Bob resource, $U_{KA}$ is simply a local unitary on the $(KA)\mid B$ bipartition, so it cannot change the Schmidt spectrum of the state across that cut \cite{NielsenChuang2000}. In particular, starting from $\ket{0}_K\otimes\ket{\psi(\theta)}_{AB}$, the Schmidt coefficients across $(KA)\mid B$ remain $\{\cos\theta,\sin\theta\}$, and deterministic restoration to a Bell pair is impossible unless $\theta=\pi/4$ (equivalently, the spectrum already matches) \cite{nielsen1999conditions}.

If instead the Bank and Bob pre-share an auxiliary entangled state (e.g., $\ket{\phi}_{KB'}$), then $U_{KA}$ effectively gives Alice access to a nonlocal subsystem that is entangled with Bob. This reduces to the asymmetric ``Bank transfers a qubit to Alice'' situation analyzed earlier: deterministic restoration is possible exactly when the consumed auxiliary resource makes the relevant majorization constraints feasible.

\subsection{Sequential dual-sided routing with a single ancilla}

If the Bank can interact sequentially with both nodes, it can act as an entanglement router even with a single-qubit ancilla $K$ initially prepared in $\ket{0}_K$. Concretely, for any $\theta\in(0,\pi/4]$ there exist unitaries $U_{KA}$ and $U_{KB}$ such that
\begin{equation}
\begin{aligned}
U_{KB}U_{KA}&\Bigl(\ket{0}_K\otimes(\cos\theta\ket{00}+\sin\theta\ket{11})_{AB}\Bigr)\\
&=\ket{0}_K\otimes\ket{\Phi^+}_{AB}
\end{aligned}
\end{equation}

One explicit construction is as follows. Choose $U_{KA}$ so that
\begin{align}
U_{KA}\ket{00}_{AK}&=\ket{\Phi^+}_{AK} \\
U_{KA}\ket{10}_{AK}&=\ket{\Psi^+}_{AK}
\end{align}
which yields the intermediate state
\begin{equation}
\ket{\Psi_1}=\cos\theta\ket{\Phi^+}_{AK}\ket{0}_B+\sin\theta\ket{\Psi^+}_{AK}\ket{1}_B
\end{equation}
Next choose $U_{KB}$ to map the orthonormal pair
\begin{align}
\cos\theta\ket{00}_{KB}+\sin\theta\ket{11}_{KB}&\mapsto \ket{00}_{KB},\\
\cos\theta\ket{10}_{KB}+\sin\theta\ket{01}_{KB}&\mapsto \ket{01}_{KB}
\end{align}
which is always possible by extending these two vectors to a basis of $\mathbb{C}^4$ \cite{NielsenChuang2000}. Applying $U_{KB}$ then produces
\begin{equation}
\ket{\Psi_2}=\ket{\Phi^+}_{AB}\otimes\ket{0}_K
\end{equation}
so the Bank restores a perfect Bell pair deterministically without distributing any pre-entangled multipartite state.
\section{Conclusion}

We developed a geometric and resource-theoretic framework for Bank-assisted teleportation through a non-maximally entangled link $\ket{\psi(\theta)}_{AB}=\cos\theta\ket{00}+\sin\theta\ket{11}$. The induced teleportation channel contracts the Bloch ball to a prolate spheroid, yielding a geometric obstruction to deterministic unit-fidelity teleportation unless $\theta=\pi/4$ (Theorem~\ref{thm:geometric_impossibility}).

Our main deterministic Bell-pair restoration results can be summarized as follows.
\begin{enumerate}
\item GHZ-class assistance admits deterministic restoration exactly when Theorem~\ref{thm:ghz_feasible} holds; moreover, for this class the Bank-measures and transfer models are equivalent.
\item W-class assistance exhibits an operational separation: the Bank-measures feasibility condition is given by Theorem~\ref{thm:w_bankmeasures_feasible}, while the transfer feasibility condition is given by Proposition~\ref{prop:w_transfer_feasible}, and the two models can differ (Corollary~\ref{cor:w_separation}).
\end{enumerate}

Beyond deterministic feasibility, we derived finite-shot optimal success probabilities for probabilistic restoration and illustrated the resulting performance curves. We also discussed quantitative resource tradeoffs between the Bank-measures and transfer models (Sec.~\ref{sec:extensions}).

Finally, we compared Bank-supplied multipartite assistance with two alternative mechanisms. Bipartite entanglement catalysis cannot deterministically overcome the $\theta<\pi/4$ barrier, but admits an optimal stochastic restoration with success probability $P_{\max}=2\sin^2\!\theta$. In a strictly stronger extension where the Bank can enact non-local unitaries, sequential dual-sided routing with a single ancilla allows deterministic restoration without distributing pre-entangled multipartite states. Further extensions are outlined in Sec.~\ref{sec:further_extensions}.

\section{Further Extensions}
\label{sec:further_extensions}

We collect several extensions suggested by the present framework.

\subsection{Network generalization: repairing and routing entanglement on graphs}
A natural extension is to replace the single imperfect edge $AB$ by a network graph whose edges correspond to imperfect entangled links. The Bank may distribute a multipartite state across many nodes and either (i) perform measurements and broadcast outcomes (Bank-measures model), or (ii) transfer subsystems to selected nodes (transfer model) to enable edge-specific Bell-pair restoration.

Two concrete directions are:
\begin{enumerate}
\item \emph{Sequential multi-edge repair:} determine when a single Bank resource can restore Bell pairs on multiple edges $A_1B_1, A_2B_2,\dots$ in sequence, possibly with degraded success probability as the resource is consumed.
\item \emph{Path and tree topologies:} characterize how Bank-measures vs transfer inequivalence behaves on repeater-like paths, where intermediate measurements induce branch-dependent corrections, and whether transferring subsystems can systematically reduce required broadcast.
\end{enumerate}

\subsection{Mixed imperfect links and noisy Bank resources}
In realistic networks, the imperfect link is mixed (e.g., Werner-type noise or amplitude damping), and the Bank resource may also be noisy. For mixed resources, pure-state majorization no longer applies directly. Two viable approaches are:
\begin{itemize}
\item \emph{Singlet-fraction bounds:} use the maximal singlet fraction to upper-bound achievable deterministic teleportation fidelity, and ask what Bank assistance can raise this bound above a target threshold.
\item \emph{SDP formulations:} express ``restore a Bell pair with probability $p$'' as a semidefinite program over separable/LOCC-relaxed operations, giving computable upper bounds and robustness certificates for the W-class separation under noise.
\end{itemize}

\subsection{Alternative operational targets: restoring channel properties}
Instead of demanding an exact Bell pair, one can target properties of the induced teleportation channel $\mathcal{E}_{\rho_{AB}}$:
\begin{enumerate}
\item \emph{Isotropy restoration:} make the Bloch ellipsoid a sphere (a depolarizing channel) even if its radius is $<1$.
\item \emph{Fidelity thresholding:} maximize the achievable teleportation fidelity subject to a fixed Bank resource budget.
\item \emph{Ellipsoid matching:} given a desired affine map (or ellipsoid), determine the minimal Bank resource needed to implement it.
\end{enumerate}
These targets connect directly to the geometric quantities introduced earlier and can be studied with channel monotones and SDP bounds.

\subsection{Sharpening the separation phenomenon}
Beyond the existence of separation, one can quantify its size and robustness. For the symmetric slice $\alpha=\gamma$, the set of $\beta^2$ values that are transfer-feasible has Lebesgue measure $\tan^2\!\theta$ (the length of the interval in Eq.~\eqref{eq:w_transfer_bound}). When $\tan^4\!\theta<1/2$, the Bank-measures model has no feasible $\beta^2$ on this slice, so the entire transfer-feasible interval corresponds to strict separation. More generally, one can ask whether separation persists for larger $\dim K$, for different Bank measurement bases, and for which resource families the two models become equivalent.

\appendix
\section{Vidal's theorem specialized to Bell-pair restoration}
\label{app:vidal}

Vidal's theorem characterizes the optimal success probability for converting a bipartite pure state $\ket{\psi}$ to another bipartite pure state $\ket{\phi}$ under LOCC \cite{vidal1999entanglement}. Let $\boldsymbol{\lambda}=(\lambda_1,\lambda_2,\dots)^{\downarrow}$ and $\boldsymbol{\mu}=(\mu_1,\mu_2,\dots)^{\downarrow}$ denote the vectors of squared Schmidt coefficients of $\ket{\psi}$ and $\ket{\phi}$, sorted in non-increasing order. Define the tail sums
\begin{equation}
E_k(\psi)=\sum_{i=k}^{d}\lambda_i\qquad E_k(\phi)=\sum_{i=k}^{d}\mu_i
\end{equation}
Then
\begin{equation}
P_{\max}(\psi\to\phi)=\min_k\frac{E_k(\psi)}{E_k(\phi)}
\end{equation}

For Bell-pair restoration we take the target to be $\ket{\phi}=\ket{\Phi^+}_{AB}\otimes\ket{\mathrm{junk}}$ for an arbitrary leftover pure state. Across the relevant bipartition the squared Schmidt vector of $\ket{\phi}$ is therefore
\begin{equation}
\boldsymbol{\mu}=\left(\frac{1}{2},\frac{1}{2},0,\dots\right)
\end{equation}
so $E_1(\phi)=1$ and $E_2(\phi)=1/2$. For the source state $\ket{\psi}$, $E_1(\psi)=1$ and $E_2(\psi)=1-\lambda_1$, where $\lambda_1=\lambda_{\max}$ is the largest squared Schmidt coefficient. Thus Vidal's formula reduces to
\begin{equation}
P_{\max}=\min\bigl(1,2(1-\lambda_{\max})\bigr)
\end{equation}
which is Eq.~(17) in the main text.

\section{Birkhoff-von Neumann Decomposition and POVM Construction}
\label{app:birkhoff}

Let $\vec\lambda=(\lambda_1,\lambda_2,\lambda_3,\lambda_4)^T$ denote the squared Schmidt coefficients of the source state across the relevant bipartition, and let $\vec\lambda_{\mathrm{target}}=(1/2,1/2,0,0)^T$ denote the target Bell-pair spectrum. If $\vec\lambda\prec\vec\lambda_{\mathrm{target}}$, then there exists a doubly stochastic matrix $D$ such that
\begin{equation}
\vec\lambda = D\,\vec\lambda_{\mathrm{target}}
\end{equation}
by the Hardy--Littlewood--P\'{o}lya characterization of majorization \cite{marshall2011inequalities}.

The Birkhoff--von Neumann theorem states that any doubly stochastic matrix is a convex combination of permutation matrices \cite{birkhoff1946tres}. Hence we may write
\begin{equation}
D = \sum_k p_k P_k,\qquad p_k\ge 0,\quad \sum_k p_k=1
\end{equation}
and therefore
\begin{equation}
\vec\lambda = \sum_k p_k P_k\,\vec\lambda_{\mathrm{target}}
\end{equation}
This decomposition directly yields a finite-outcome LOCC protocol: each permutation $P_k$ corresponds to a local unitary correction, and the weights $p_k$ determine the POVM branch probabilities.

\par\medskip
\refstepcounter{algorithm}
\noindent\textbf{Algorithm~\thealgorithm. Constructive Birkhoff--von Neumann decomposition}\label{alg:bvn}
\par\smallskip
\begin{algorithmic}[1]
\Require Doubly stochastic matrix $D\in\mathbb{R}^{d\times d}$
\Ensure A decomposition $D=\sum_k p_k P_k$ into permutation matrices
\State $R\gets D$ \Comment{Residual matrix}
\While{$R\neq 0$}
    \State Build a bipartite graph $G=(U,V,E)$ with $(i,j)\in E$ iff $R_{ij}>0$
    \State Find a perfect matching $M$ in $G$ (guaranteed for doubly stochastic $R$ by Hall's theorem \cite{hall1935representatives})
    \State Let $P$ be the permutation matrix corresponding to $M$
    \State $p\gets \min\{R_{ij}:(P)_{ij}=1\}$
    \State Output $(p,P)$
    \State $R\gets R-pP$
\EndWhile
\end{algorithmic}
\par\medskip

\bibliography{main}

\end{document}